\documentclass[11pt]{article}

\usepackage{amsmath,amssymb,amsthm,amsfonts,mathptmx,mathrsfs,fullpage}
\usepackage{tablefootnote}
\usepackage[margin=1in]{geometry}

\usepackage{hyperref} 

\usepackage{graphicx,color,cases} 
\usepackage{xfrac}  
\definecolor{DarkGreen}{rgb}{0.1,0.5,0.1}
\definecolor{DarkRed}{rgb}{0.5,0.1,0.1}
\definecolor{DarkBlue}{rgb}{0.1,0.1,0.5} 
 
\newtheorem{theorem}{Theorem} 
\newtheorem{lemma}[theorem]{Lemma} 
\newtheorem{corollary}[theorem]{Corollary}
 
\newtheorem{proposition}[theorem]{Proposition}
\theoremstyle{definition} 

\newtheorem{example}[theorem]{Example}

\numberwithin{equation}{section}
 
\def\>{\rangle} 
\def\<{\langle}

\newcommand{\vzero}[1]{{\bf 0}_{#1} }

\DeclareMathOperator{\spann}{span}

\DeclareMathOperator{\wt}{wt}

\title{Permutation-invariant constant-excitation quantum codes for amplitude damping}
\author{Yingkai Ouyang$^1$ \and Rui Chao$^2$}
\date{%
    $^1$ University of Sheffield, UK\\%
    $^2$ University of Southern California, USA\\[2ex]%
    \today
}

\begin{document}

\maketitle 

\abstract{
 The increasing interest in using quantum error correcting codes in practical devices has heightened the need for designing quantum error correcting codes that can correct against specialized errors, such as that of amplitude damping errors which model photon loss.
 Although considerable research has been devoted to quantum error correcting codes for amplitude damping, not so much attention has been paid to having these codes 
 simultaneously lie within the decoherence free subspace of their underlying physical system.
 One common physical system comprises of quantum harmonic oscillators, 
 and constant-excitation quantum codes can be naturally stabilized within them.
 The purpose of this paper is to give constant-excitation quantum codes that not only correct amplitude damping errors, but are also immune against permutations of their underlying modes.
To construct such quantum codes, we use the nullspace of a specially constructed matrix based on integer partitions.
}

\section{Introduction}
 
The ability to manipulate quantum information promises to speed up algorithms such as factoring \cite{shor1994algorithms,shor1999polynomial},
simulate physical systems more efficiently \cite{lloyd1996universal}, 
and unlock the ability to perform cryptographic schemes with unprecedented security \cite{BB84,Eke91}.
However the inherent fragility of quantum information is a major obstacle in realizing the full potential of these quantum schemes.
To overcome this, one may rely on quantum error correction codes, which offer the possibility of reversing the effects of decoherence \cite{KnL97}. 
However, if an arbitrary quantum error correction code were used in a physical system,
it is invariably affected by the underlying system's natural dynamics.
Therein lies the allure of constructing quantum error correction codes within an energy eigenspace of the physical system's underlying Hamiltonian, because errors may then be avoided at a fundamental level.

We consider here the problem of quantum error correction in quantum harmonic oscillators, which, as we shall see, is not just of theoretical interest.
In recent years, superconducting qubits have been extensively studied, and are considered as one of the leading candidates for realizing quantum information in a physical system.
In the superconducting electrical circuits that superconducting qubits are based on, each superconducting qubit is localized around a cluster of Josephson junctions, 
and can interact with other spatially separated superconducting qubits when coupled with microwave-frequency photons in a quantum bus \cite{blais2004cavity,majer2007coupling,gu2017microwave}.
It is well-known that this quantum bus is in turn just a microwave-frequency electrical transmission line \cite{gu2017microwave}[Appendix A.3], with a Hamiltonian 
described by a sum of quantum harmonic oscillators.

If we restrict our attention to using photons of identical frequencies within a quantum bus, 
its Hamiltonian is up to a constant effectively given by $H = \sum_{j}a_j^\dagger a_j$.
Here, $a_j$ denotes the lowering operator for the $j$-th mode.
Now let $|x_1\> \otimes \dots \otimes |x_n\>$ denote a quantum state with $x_j$ excitations in the $j$-th mode,
and let $x_1+ \dots + x_n$ denote the total excitation number of such a state. 
Then the eigenspaces of $H$ are spanned by states $|x_1\> \otimes  \dots \otimes |x_n\>$ with a constant total excitation number. 
If a quantum code is spanned by states with a constant total excitation number, 
we call it a constant-excitation quantum code. 
In this paper, we design constant-excitation quantum codes which can be stabilized by the Hamiltonian of quantum harmonic oscillators.

We consider two common types of errors that may afflict our physical system modeled by quantum harmonic oscillators, particularly in a quantum bus. 
The first type of errors are amplitude damping (AD) errors, which can arise 
when the system weakly interacts with a zero temperature bosonic and Markovian bath.
The second type of errors are permutation errors, 
which may arise especially during transmission when modes are unexpectedly permuted either spatially or temporally. 
In this paper, we consider quantum error correction codes that offer protection against not only amplitude damping errors but also permutation errors.

Amplitude damping errors model energy relaxation in quantum harmonic oscillator systems and photon loss in photonic systems. In this paper, we will consider amplitude damping errors that occur independently and identically on every mode.
To see how amplitude damping errors may arise in quantum harmonic oscillators from the underlying physics, consider the coupling Hamiltonian on the $j$-th mode given by 
\[
H_{{\rm int},j} = \chi_j (a_j ^\dagger b_j +b_j^\dagger a_j).
\]
Here, $b_j$ is the lowering operator of the bath that couples to the $j$-th mode,
and so $H_{{\rm int},j}$ couples each unique quantum harmonic oscillator to a unique bath, thereby ensuring that the amplitude damping errors occur independently for each mode.
There are two reasons why we do not consider couplings between the harmonic oscillators. First the Hamiltonian of an ideal transmission line naturally comprises of uncoupled harmonic oscillators. Second, even if there is some spurious linear coupling between the harmonic oscillators, it has later been shown that dynamical decoupling can homogenize the system, and render the effective Hamiltonian to be that of a sum of identical uncoupled harmonic oscillators \cite{heinze2018universal}.
In this paper, we assume that AD errors afflict every mode identically, which can be the case when the coupling strength $\chi_j$ is independent of $j$, so that $\chi_j = \chi$.
By assuming that the system and the bath are initially in a product state, and subsequently applying a Born-Markov approximation \footnote{See \cite{leung-thesis} for a detailed exposition.}, one can show that the noise process can be modeled using the Kraus operators
\begin{align}
A_k = \sum_{m=k}^\infty \sqrt{\binom{m}{k}}\sqrt{(1-\gamma)^{m-k} \gamma^k} |m-k\>\<m|,
\end{align}
where $k$ indicates the number of AD errors that afflict a mode, and $\gamma = 1-\cos^2(\chi  \Delta t)$ is the strength of the AD error that corresponds to time $\Delta t$ \cite{CLY97}.

Permutation errors model the stochastic reordering and coherent exchange of quantum packets as well as out-of-order delivery of packets of information, which plausibly occur due to imperfections in a communication channel \cite{ouyang2015permutation}.
More precisely, we denote a permutation channel to be a quantum channel with each of its Kraus operators $P_\alpha$ proportional to $e^{i \theta_\alpha \hat \pi_\alpha}=\sum_{k \ge 0} (i \theta_\alpha \hat \pi_\alpha)^k/k!$, where $\theta_\alpha$ is the parameter corresponding to the infinitesimal generator $i \hat \pi_\alpha$ \footnote{The matrix $i \hat \pi_\alpha$ is called an infinitesimal generator because it is a bounded operator, and generates the unitary matrix $e^{i \theta_\alpha \hat \pi_\alpha}$ in the sense consistent with \cite[Theorem 1.2]{pazy2012semigroups}. }, and $\hat \pi_\alpha$,  is any linear combination of operators that permute the underlying modes with real coefficients.
We also emphasize that we are interested in the scenario where $\theta_\alpha$ takes an arbitrary value from the real numbers, so it need not be small.   

Amplitude damping errors are prevalent in bosonic systems.
If a single bosonic mode were left unprotected against AD errors, the incurred error as quantified by one minus the fidelity is of order $\gamma$.
If a quantum error correction code allows reduction of the order of this error to $\gamma^{t+1}$, we say that the quantum code corrects $t$ AD errors.
Unsurprisingly, there has been extensive work on quantum error correction codes specialized against correcting amplitude damping errors \cite{LNCY97,CLY97,cochrane1999macroscopically,SSSZ09, DGJZ10,BinomialCodes2016,jackson2016concatenated,li2017cat,grassl2018quantum}.
Of special note are some previously constructed constant-excitation quantum codes that do offer immunity against the natural dynamics of quantum harmonic oscillators \cite{CLY97,WaB07,BvL16}.
Chuang, Leung and Yamamoto restricted their study of bosonic quantum codes to constant-excitation quantum codes, 
and found that the fidelity after quantum error correction for such codes that correct $t$ AD errors with total excitation number $N$ can be made to be 
$\sum_{k=0}^t \binom N k \gamma^k (1-\gamma)^{N-k} = 1 - \binom{N}{t+1} \gamma^{t+1} + O(\gamma^{t+2})$ \cite{CLY97}.
Hence for constant-excitation quantum codes, minimizing $N$ for fixed $t$ is the primary goal. 
In their paper \cite{CLY97}, Chuang, Leung and Yamamoto 
also found constant-excitation quantum codes correcting 1, 2 and 3 AD errors with total excitation numbers equal to 4, 9 and 16 respectively.
Wasilewski and Banaszek later introduced a constant-excitation quantum code with $N=3$ correcting 1 AD error \cite{WaB07}, thereby improving on the construction of the $N=4$ code in \cite{CLY97}. 
The code of Wasilewski and Banaszek is also notably the first known constant-excitation quantum code that not only corrects amplitude damping errors, but is also permutation invariant.
Recently, Bergmann and van Loock found constant-excitation quantum codes that can correct any number of AD errors \cite{BvL16}. Namely, their codes can correct $t$ AD errors using $N=(t+1)^2$ excitations and are very elegant in the sense that these codes can be encoded simply by using NOON states and beamsplitters.
Unfortunately the codes of Bergmann and van Loock are not invariant under arbitrary permutations, and are hence vulnerable to certain permutation errors.

  Apart from specialized quantum codes that correct amplitude damping errors, 
  permutation-invariant quantum codes have also been studied in recent years, 
  both with respect to arbitrary errors \cite{Rus00,PoR04,ouyang2014permutation,OUYANG201743} and also amplitude damping errors \cite{WaB07,ouyang2014permutation,ouyang2015permutation}.
  Permutation-invariant quantum codes are important because they are inherently immune to permutation errors.
After Ruskai first introduced a 9 qubit permutation-invariant quantum code correcting one arbitrary error \cite{Rus00},
Pollatsek and Ruskai later improved this in \cite{PoR04} with a 7 qubit permutation-invariant code that corrects one arbitrary error. 
Later in \cite{ouyang2014permutation}, Ruskai's 9 qubit permutation-invariant quantum code was generalized to yield permutation-invariant quantum codes correcting $t$ arbitrary errors or $t$ AD errors while encoding a single qubit. 
  In \cite{ouyang2015permutation} and \cite{OUYANG201743}, the permutation-invariant codes were generalized in different directions to allow the correction of 1 AD error and encoding of a qudit, and correction of arbitrary errors and encoding of a qudit respectively.  
  However, aside from Wasilewski and Banaszek's quantum code, none of these permutation-invariant codes are also constant-excitation quantum codes.
 
The purpose of this paper is to construct constant-excitation quantum codes that not only correct any $t$ AD errors, but are also permutation-invariant (PI). 
Using the techniques from linear algebra and by counting the sizes of integer partitions, we construct PI constant-excitation quantum codes that correct $t$ AD errors for any integer $t$.
For our codes, the total number of modes used is equal to the total excitation number, so $n=N$.
For example, when the total excitation number $N$ satisfies the following inequality,
\begin{align}
p\left(\frac N{t+1} \right) + \binom t 2 \ge  p(1) + \dots + p(t) \label{eq:inequality}
\end{align}
where $p(t)$ denotes the number of integer partitions of $t$,
there are corresponding PI constant-excitation quantum codes that correct $t$ AD errors.
The inequality in (\ref{eq:inequality}) allows us to easily find code parameters for PI constant-excitation quantum codes.

Among the PI constant-excitation quantum codes that we construct, 
we have codes that correct 2,3,4 and 5 AD errors using 6,12, 20 and 30 total excitations respectively. 
These codes are given explicitly in Example \ref{example-t=1}, Example \ref{example-t=2}, Example \ref{example-t=3,N=12}, Example \ref{example-t=4,N=20} and Example \ref{example-t=5,N=30} respectively.
We wish to emphasize that these codes are not only permutation-invariant, 
but also have lower total excitation numbers than the constant-excitation codes of Bergmann and van Loock, which require 9,16,25 and 36 total excitations respectively.
In this sense, for small values of $t$, our constructed codes give the best performance in terms of fidelity among the constant-excitation quantum codes.
Moreover for large values of $t$, we numerically find that our constructed PI constant-excitation quantum codes that correct $t$ AD errors have total excitations $N = C (t+1)^2$ where $C$ is slightly larger than one (see Figure \ref{fig:bounds}).
This suggests that our code parameters are asymptotically similar to those of Bergmann and van Loock. 
The value of this result lies in the fact that 
permutation-invariance can be imbued to constant-excitation quantum codes while minimally affecting their output fidelities.

In this paper, we construct our PI constant-excitation quantum code using 
partitions of a well-chosen integer and a real vector given explicitly in (\ref{eq:code-defi}). 
Independently, we define a matrix $A$ in (\ref{eq:aij}) that depends only on the partitions that label the AD errors that are to be corrected and the partitions that label the permutation-invariant states of constant-excitation that our code is to be supported on.
We prove in Theorem \ref{thm:nullity-2-code} that any non-trivial solution to the linear system of equations $A{\bf x} = {\bf 0}$ leads to a PI constant-excitation quantum code.
Intuitively, the matrix $A$ quantifies the extent in which AD errors, after acting on Dicke states, can shrink their norms.
By obtaining a lower bound on the nullity of $A$, we prove (\ref{eq:inequality}) in
Corollary \ref{coro}.
 
\section{Preliminaries and notation}\label{sec:notations}
We begin by introducing terminology related to vectors of non-negative integers.
First define $\mathbb N$ to be the set of non-negative integers and let $n$ be a positive integer denoting the number of modes and the total excitation number that will be used for the quantum code. 
For any integer $a$ and non-negative integer $b$, let $a_{(b)} = (a)(a-1) ... (a-b+1)$ denote the falling factorial symbol.
Here, $(a)_{(0)}=1$.
Let $(y_1, \dots, y_n) $ denote a column vector and $(y_1, \dots, y_n) ^T$ denote a row vector.
Define ${\bf 1}_u$ and ${\bf 0}_u$ as column vectors of length $u$ with all components equal to 1 and 0 respectively. 
For ${\bf y} = (y_1, \dots, y_n)$, let $\wt({\bf y}) = y_1 + \dots + y_n$ denote the weight of ${\bf y}$.
For non-negative integers $t$ such that $0 \le t\le n$, 
let $\mathcal K_{n,t} = \{ (y_1, \dots, y_n) \in \mathbb N^n: y_1 + \dots + y_n = t  \}$
denote the set of non-negative vectors of weight $t$.
Also define $\overline {\mathcal K}_{n,t} = \mathcal K_{n,0} \cup \dots \cup \mathcal K_{n,t}$
to be the set of non-negative vectors with weights from 0 to $t$.

We now introduce terminology related to the constant-excitaton quantum codes that we will study.
Let the orthonormal vectors $|j\>$ for $j \in \mathbb N $ span the Hilbert space of a single bosonic mode, which we denote as $\mathcal H$.  
The quantum codes that we consider in this paper are two-dimensional subspaces of the $n$-mode Hilbert space $\mathcal H_n$.
Given a vector ${\bf y} = (y_1, \dots, y_n) \in \mathbb N^n$, 
define the computational basis state $|{\bf y}\> = |y_1\> \otimes \dots \otimes |y_n\> \in \mathcal H_n$.
The weight, or a total excitation number of a computational basis state $|{\bf y}\>$ is the weight of 
${\bf y}$.
We say that a quantum code is also a constant-excitation quantum code if it 
can be spanned by linear combinations of states with a constant total excitation number.

In this paper, we deal with the matrices $A_k ^\dagger A_k$ repeatedly, and hence we evaluate them first.
\begin{proposition}
\label{prop:A-dag-A}
For all non-negative integers $k$, we have $ A_k^\dagger A_k = \sum_{j = k}^\infty \binom{j}{k} (1-\gamma )^{j-k}\gamma^k |j\>\<j|$.
\end{proposition}
We now require notation for representing AD errors that occur on $n$ modes. Given a vector
${\bf k} = (k_1,\dots , k_n)  \in \mathcal K_{n,\kappa} $, 
let $A_{\bf k} = A_{k_1} \otimes \dots \otimes A_{k_n}$.
We say that $A_{\bf k} $ has a weight of $\kappa$.
We can then find that the diagonal matrix elements of 
$\sum_{{\bf k} \in \mathcal K_{n,\kappa}}  A_{\bf k} ^\dagger A_{\bf k}$ in the computational basis 
$|x_1\> \otimes \dots \otimes |x_n\>$ only depends on $\kappa$ and $x_1 +\dots + x_n$. 
The following proposition, which essentially follows the same logic as the equations from (7.6) to (7.11) in \cite{CLY97}, makes this precise.
\begin{proposition} \label{prop:const}
Let ${\bf x} =(x_1,\dots, x_n)$ be a vector of non-negative integers, and let $\chi = x_1 + \dots +x_n$. 
Then
\begin{align}
\sum_{\substack{
{\bf k}  \in  \mathcal K_{n,\kappa} \\ 
}}
 \< {\bf x} | A_{{\bf k}} ^\dagger A_{{\bf k}}  | {\bf x} \>
=
  (1-\gamma )^{\chi-\kappa}\gamma^{\kappa} 
   \binom{ x_1 + \dots + x_n }{\kappa} .
   \notag
\end{align}
\end{proposition} 
\begin{proof}
Obviously, $\< {\bf x} | A_{{\bf k}} ^\dagger A_{{\bf k}}  | {\bf x} \>=  \<x_1 | A_{k_1}^\dagger  A_{k_1} | x_1 \> \dots \<x_n | A_{k_n}^\dagger  A_{k_n} | x_n \>  $.
Using Proposition \ref{prop:A-dag-A}, we get
\begin{align}
\< {\bf x} | A_{{\bf k}} ^\dagger A_{{\bf k}}  | {\bf x} \>=
\prod_{i=1}^n \binom{x_i}{k_i} (1-\gamma )^{x_i-k_i}\gamma^{k_i}  
= 
(1-\gamma )^{\chi-\kappa}\gamma^{\kappa}  \prod_{i=1}^n \binom{x_i}{k_i} \notag.
\end{align}
Note that this equality holds even when $x_i < k_i$ for some $i$, because the above equality does show that $\< {\bf x} | A_{{\bf k}} ^\dagger A_{{\bf k}}  | {\bf x} \>$ is zero in this case, which is a simple consequence of the fact that $\binom {x_i}{k_i}  = (x_i)_{(k_i)}/ k_i! = (x_i) \dots (0) \dots (x_i-k_i+1) / k_i! = 0$ whenever $x_i < k_i$.
The result follows conditioned on the combinatorial identity
\begin{align}
\sum_{\substack{{\bf k}   \in \mathcal K_{n,\kappa} }}
\prod_{i =1}^n \binom{x_i}{k_i} = \binom{x_1+\dots +x_n}{\kappa}, \label{eq:combi-id}
\end{align}
which is stated in \cite[(7.10)]{CLY97} for example.
We can see that (\ref{eq:combi-id}) holds combinatorially, since the number of ways to select $\kappa$ balls from bags with $x_1, \dots, x_n$ balls is sum of the product of the number of ways to select $k_i$ balls from the $i$th bag, where the total number of selected balls is $\kappa$. 
Alternatively, we can also prove (\ref{eq:combi-id}) algebraically using the method of generating functions.
Given any formal power series $f(z) = \sum_{j\ge 0}f_j z^j$ where $z$ is indeterminate, let $[z^k]f(z) = f_k$.
Then we can write $\binom {x_i}{k_i} = [z_i^{k_i}](1+z_i)^{x_i}$ for indeterminates $z_i$.
Hence, we can rewrite the left hand side of (\ref{eq:combi-id}) as 
\begin{align}
\sum_{\substack{
{\bf k}   \in \mathcal K_{n,\kappa}  \\ 
}} 
[z_1^{k_1}](1+z_1)^{x_1} \dots  [z_n^{k_n}](1+z_n)^{x_n}
&=[z^\kappa](1+z_1)^{x_1} \dots (1+z_n)^{x_n} \notag\\
&=[z^\kappa](1+z_1)^{x_1+  \dots+ x_n} \notag\\
&=\binom {x_1 +\dots + x_n}{\kappa},\notag
\end{align}
which proves (\ref{eq:combi-id}).
\end{proof}
Because the quantum codes that we study are invariant under any permutation of the underlying modes, we proceed to define some permutation-invariant states that we will use to construct our quantum codes. 
To do so, we have to first introduce notation related to integer partitions.
Given a positive integer $n$, let $p(n)$ denote the number of its partitions, and denote every partition of $n$ as an $n$-tuple of non-increasing non-negative integers. For example with $5=2+2+1$, we denote the corresponding partition of 5 as $(2,2,1,0,0)$. 
We denote the set of partitions of $n$ with $P(n)$.  
Given tuples ${\bf x} = (x_1, \dots, x_a)$ and ${\bf y}  = (y_1,\dots, y_b),$
let $({\bf x} | {\bf y}) = (x_1,\dots x_a, y_1, \dots, y_b)$ denote the pasting of tuples ${\bf x}$ and ${\bf y}$.
For all positive integers $t$, let 
\begin{align}
\overline P(t) = \{ ({\bf q} | {\bf 0}_{t-k} ) : {\bf q} \in P(k) , k = 1,\dots, t-1 \}  \cup P(t),
\end{align}
and let
\begin{align}
\overline p(t) = p(1) + \dots + p(t) = |\overline P(t)|.
\end{align}
Given any positive integer $w$, let $Q$ be an arbitrary subset of partitions $P(w)$, {\em i.e} $Q \subseteq P(w)$.
The integers $u$ and $w$ are such that the number of modes $n$ satisfies the constraint $n=uw$.
Furthermore, given a positive integer $u$, let $Q_u$ be the set of vectors in $Q$ multiplied by $u$ component-wise and extended to length $uw$, and appended by a ones vector, so that we have
\begin{align}
Q_u = \{ ( u {\bf q}  | {\bf 0}_{uw-w}  ):   {\bf q}\in Q\}  \cup \{{\bf 1}_{n}\}. \label{eq:Qu-defi}
\end{align}
We will see later how our quantum codes can be constructed using linear combinations of states labeled by vectors in $Q_u$.
For example, the ones vector in $Q_u$ will correspond to the permutation-invariant state $|1\>^{\otimes uw}$.
We proceed to illustrate this definition with some examples.
\begin{enumerate}
\item When $Q =P(1) = \{ (1) \}$, we have $Q_3 = \{(3,0,0) , (1,1,1)\}$. 
The set $Q_3$ will later correspond to the permutation-invariant basis 
\begin{align}
\left \{ |1,1,1\> , \frac{|3,0,0\> + |0,3,0\> + |0,0,3\> }{\sqrt 3} \right\}.
\end{align}
\item When $Q = P(2) = \{(2,0),(1,1)\} $, 
we have
$Q_{3} = \{ (6,0,0,0,0,0), (3,3,0,0,0,0)\}$.
The set $Q_3$ will later correspond to the permutation-invariant basis $\{ |1,1,1,1,1,1\> , |\phi_1\>,|\phi_2\>\} $ where
\begin{align}
\sqrt 6|\phi_1\> &=  | (6| \vzero 5)\> + |(0,6| \vzero 4)\> + |(\vzero 2|6|\vzero 3)\> 
+|(\vzero 3|6|\vzero 2)\> + |(\vzero 4| 6,0)  \> + |(\vzero 5| 6) \>    \notag\\
\sqrt{15} |\phi_2 \> &= |(3,3|\vzero 4)\> + |(0,3,3|\vzero 3)\> + |(\vzero 2|3,3|\vzero 2)\> 
+|(\vzero 3|3,3,0)\> + |(\vzero 4 |3,3 ) \> + |(3|\vzero 4|3) \>   \notag\\
&\quad 
+| (3,0,3|\vzero 3) \>
+| (0,3,0,3|\vzero 2) \>
+| (\vzero 2 |3,0,3,0) \>
+| (\vzero 3|3,0,3 ) \>
+| (3|\vzero 3|3,0 ) \>
+| (0,3|\vzero 3|3 ) \>\notag\\
&\quad 
+| (3,0,0,3,0,0) \>
+| (0,3,0,0,3,0) \>
+| (0,0,3,0,0,3) \>.
\end{align}
\end{enumerate}
Given a partition ${\bf q}=(q_1, \dots, q_n)$ of a positive integer $n$, we first define $\widetilde{\bf q}$ as the set of all permutations of ${\bf q}$.
For any set of non-negative vectors $Q$, we define $\widetilde Q$ to be a union of all the sets $\widetilde{\bf q}$ for which ${\bf q} \in Q$.
Formally, denoting $S_n$ as the symmetric group of order $n$, we have
\begin{align}
\widetilde{\bf q} &= \{  (q_{\pi(1)}, \dots, q_{\pi(n)}: \pi \in S_n\},\\
\widetilde Q_u &= \bigcup_{q \in Q_u} \widetilde {\bf q}.\label{def:widetilde Qu}
\end{align}
For example,
\begin{align}
\widetilde{(1,1,1,0)}  &=  \{ (1,1,1,0), (1,1,0,1),(1,0,1,1),(0,1,1,1) \}   . \notag
\end{align} 
Next we define $|\widetilde{\bf q}\>$ as a uniform superposition over all permutations of the $n$-mode state 
$|q_1\> \otimes \dots \otimes |q_n\>$,
so that 
\begin{align}
|\widetilde{\bf q}\> = \frac{1}{\sqrt{ |\widetilde{\bf q} |  } } \sum_{ {\bf y} \in \widetilde{\bf q} } |{\bf y}\>.
\end{align}
As an example using this notation, 
\begin{align}
|\widetilde{(2,2,0,0)}\> &=  \frac{1}{\sqrt{6} }(
|(2,2,0,0)\> + |(0,2,2,0)\> + |(0,0,2,2)\> + |(2,0,0,2)\> +|(2,0,2,0)\> + |(0,2,0,2)\> )  .\notag
\end{align}
The codes that we consider lie within the span of these Dicke states $|\widetilde{\bf q}\>$ where ${\bf q}$ are partitions of $n$. 
In particular, $  \spann \{ |\widetilde{\bf q}\> : {\bf q} \in P(n) \}$ is the space of all permutation-invariant states with a total-excitation number $n$.

Let us denote the distance between any two vectors ${\bf u}  =( u_1,\dots, u_n)$ and ${\bf v} = (v_1,\dots, v_n)$ in $\mathbb N^n$ as 
\begin{align}
d({\bf u},{\bf v}) = \sum_{j=1}^n | u_j - v_j| \label{dist-criterion}.
\end{align}
This distance, also known as the Manhattan distance, is distinct from the usual Hamming distance.
Let $C \subseteq \mathbb N^n$. Then we define the minimum distance of $C$ to be 
\begin{align}
d(C) = \min_{ {\bf u},{\bf v} \in C } 
\{ d({\bf u},{\bf v}) :   {\bf u} \neq {\bf v}\}.
\end{align}
Our non-standard definition of minimum distance arises because we use a metric induced by the Manhattan distance as opposed to the Hamming distance.
We will later be interested in the minimum distance of the set $\widetilde{Q_u}$ for $Q \subseteq P(w)$ for some positive integer $w$.

\section{Quantum error correction criterion}
Here, we review some of the underlying theory of quantum error correction for AD errors on constant-excitation quantum codes, beginning from the Knill-Laflamme (KL) quantum error correction criterion \cite{KnL97} .
Given that we wish to correct $t$ AD errors on $n$ modes, it is both necessary and sufficient to have the KL conditions satisfied for the AD errors $A_{\bf k}$ for which $k_1+\dots + k_n \le t$.
Just as in \cite{CLY97}, for the KL quantum error correction criterion to hold for a quantum code with logical codewords $|0_L\>$ and $|1_L\>$ on $n$ modes with respect to $t$ AD errors, it suffices to require that the following equations hold.
\begin{align}
\<0_L |A_{\bf x}^\dagger A_{\bf x} |0_L\> &= 
\<1_L |A_{\bf x}^\dagger A_{\bf x} |1_L\>
 \quad \forall {\bf x} \in \overline {\mathcal K}_{n,t}
  \label{nodeform-diag}, \\
\<0_L |A_{\bf x}^\dagger A_{\bf y} |0_L\> &= 
\<1_L |A_{\bf x}^\dagger A_{\bf y} |1_L\> = 0
 \quad \forall \mbox{ distinct } {\bf x} , {\bf y} \in \overline {\mathcal K}_{n,t}
  \label{nodeform-offdiag}, \\
\<0_L| A_{\bf x}^\dagger A_{\bf y} | 1_L \> &= 0 
 \quad \forall {\bf x} , {\bf y} \in \overline {\mathcal K}_{n,t}
\label{ortho}.
\end{align}
In the usual KL quantum error correction criterion, the right hand side of (\ref{nodeform-offdiag}) does not have to be zero. 
We have made this restriction just as in \cite{CLY97} to make the construction of our quantum codes more tractable.
For convenience, we hereby call the constraints (\ref{nodeform-diag}) the non-deformation quantum error correction criterions, 
and the constraints (\ref{nodeform-offdiag}) and (\ref{ortho}) the orthogonal quantum error correction criterions.

For us, the condition (\ref{nodeform-offdiag}) always holds as long as the logical codewords are linear combinations of vectors $|{\bf v}\>$ for which ${\bf v}\in C \subset \mathbb N^{n}$ and $d(C)\ge 2t+1$.
To see this, note first that from \cite[Theorem 2]{CLY97}, as long as ${\bf v}$ and ${\bf w}$ are distinct vectors from $\mathbb N^n$ where $d({\bf v},{\bf w})\ge 2t+1$, then for every ${\bf x}$ and ${\bf y}$ in $\overline {\mathcal K}_{n,t}$, we have $\< {\bf v} |  A_{\bf x} ^\dagger A_{\bf  y} |  {\bf w}  \>=0$.
Second, when ${\bf x}$ and ${\bf y}$ are distinct elements from $\overline {\mathcal K}_{n,t}$,
it is clear that $A_{\bf x}|{\bf v}\>$ and $A_{\bf y}|{\bf v}\>$ are either zero or $\alpha_{\bf x} | {\bf u}\> 
$ and $\alpha_{\bf y} | {\bf u}'\> 
$ respectively for real numbers $\alpha_{\bf x},\alpha_{\bf y}$ and vectors ${\bf u},{\bf u}'$ in $\mathbb N^n$.
Since ${\bf u}\neq {\bf u}'$ from the distinctness of ${\bf x}$ and ${\bf y}$, it follows that 
$\<{\bf v}| A_{\bf y}^\dagger A_{\bf x}  |{\bf v}\> = 0$ for all distinct ${\bf x} $ and ${\bf y}$ in $\overline {\mathcal K}_{n,t}$. 
Hence it follows that if $|\psi\>$ is any vector that is a linear combination of basis states $|{\bf v}\>$ for which ${\bf v}\in C$ with $d(C)\ge 2t+1$, for all distinct ${\bf x} $ and ${\bf y}$ in $\overline {\mathcal K}_{n,t}$,
we must have $\<\psi | A_{\bf y}^\dagger A_{\bf x}  |\psi\>=  0$.

When the condition (\ref{nodeform-offdiag}) holds for a quantum code and $\<i_L|A_{\bf x} ^\dagger A_{\bf x}|i_L\> > 0$ for all $i=0,1$ and ${\bf x} \in \overline {\mathcal K}_{n,t}$, the quantum code is non-degenerate.
To see this, note that
\begin{align}
\sum_{{\bf x} \in \overline {\mathcal K}_{n,t}}\<i_L|A_{\bf x} ^\dagger A_{\bf x}|i_L\> |{\bf x}\>\<{\bf x}|\notag
\end{align}
is diagonal with positive diagonal entries.
Since this matrix is invertible and hence full rank, Gottesman's definition \cite[Page 7, line 1]{gottesman2002introduction},
implies that such quantum codes are non-degenerate.

Since the orthogonality condition (\ref{ortho}) holds as the logical codewords are linear combinations of vectors $|{\bf v}\>$ for which ${\bf v}\in C \subset \mathbb N^{n}$ and $d(C)\ge 2t+1$ and ${\bf x} $ and ${\bf y}$ in $\overline {\mathcal K}_{n,t}$ \cite[Theorem 2]{CLY97},
the only non-trivial error correction criterion is the non-deformation condition (\ref{nodeform-diag}).
Because quantum codes that we construct are permutation-invariant, it suffices to restrict the error-inducing Kraus operators that arise from partitions of $\kappa$ where $\kappa \le t$ and $t \le n$.
This is because for any permutation-invariant quantum state $|\psi\>$ and any Kraus operator $B$, we have
\begin{align}
\<\psi| B ^\dagger B   |\psi\> 
=  \<\psi| \pi ^\dagger  B^\dagger \pi \pi ^\dagger B \pi  |\psi\> 
=  \<\psi| (\pi ^\dagger  B \pi ) ^\dagger \pi ^\dagger B \pi  |\psi\>\notag.
\end{align}
Hence for every partition $\lambda = (\lambda_1, \dots, \lambda_\kappa)$ in $P(\kappa)$, we denote by
$A_{\lambda,n}$ the amplitude damping operator on $n$ modes with respect to ${\lambda}$ where 
\begin{align}
A_{\lambda,n} =  A_{\lambda_1} \otimes \dots \otimes A_{\lambda_\kappa} \otimes A_0 ^{\otimes n-\kappa}. 
\end{align}

If the conditions (\ref{nodeform-diag}), (\ref{nodeform-offdiag}), and (\ref{ortho}) hold for the constant-excitation quantum code with total excitation number $n$, then the worst-case fidelity is at least $\sum_{k=0}^t \binom n k \gamma^k (1-\gamma)^{n-k},$ as proved in \cite{CLY97}.
In fact, the entanglement fidelity exhibits the same behavior, as we now illustrate.

The entanglement fidelity of a quantum code quantifies how well an entangled state 
\[|\psi\> =
\frac{|0\>\otimes |0_L\> +  |1\> \otimes |1_L\>}{\sqrt 2}
\]
is protected when 
the half of it which is encoded into a quantum code 
with logical codewords $|0_L\>$ and $|1_L\>$ is exposed to noise.
If the recovery channel of the quantum code is given by $\mathcal R$, 
its entanglement fidelity with respect to AD errors is 
\begin{align}
\<\psi| \mathcal I  \otimes (\mathcal R \circ \mathcal A)(|\psi\>\<\psi|)  |\psi\>
=
\left(
\<0_L| \mathcal R (\mathcal A( |0_L\>\<0_L| ) ) |0_L\> + 
\<1_L| \mathcal R (\mathcal A( |1_L\>\<1_L| ) ) |1_L\> 
\right)/2
,
\end{align}
where $\mathcal I$ is the identity channel on a single qubit, and $\mathcal A$ is the quantum channel corresponding to an AD channel that acts independently and identically on every mode in the quantum code.
Now we can write $\mathcal A = \mathcal A' + \mathcal A''$ where $\mathcal A'$ and $\mathcal A''$ are both quantum operations that induce at most $t$ AD errors and at least $t+1$ AD errors respectively.
Clearly if the quantum code is completely correctible with respect to the quantum operation $\mathcal A'$, 
then Proposition \ref{prop:const}
implies that the entanglement fidelity is at least the trace of 
\[
 \left( \mathcal A'( |0_L\>\<0_L| )  + \mathcal A'( |1_L\>\<1_L| )  \right)/2,
\]
which is at least $\sum_{k=0}^t \binom n k \gamma^k (1-\gamma)^{n-k}$, if the quantum code is a constant-excitation quantum code with $n$ total excitations.

\section{From partitions to quantum codes}
\label{sec:t2}
Here we will see how a PI constant-excitation quantum code can be constructed from integer partitions. 
Some of the integer partitions label the AD errors, while the others label the Dicke states that our code is supported on.
It is the 
permutation-invariant
property of our code that allows us 
to restrict our attention to AD errors that are labeled by integer partitions of the numbers from 1 to $t$.
To be more explicit, since the norm of an AD error acting on a permutation-invariant state is equivalent to the norm of a permuted AD error acting on the same permutation-invariant state, in studying the non-deformation conditions, it suffices to study only the AD errors labeled by integer partitions of the number of AD errors.
We label these AD errors with 
the vectors 
$\tau_1 ,\dots,\tau_{\overline p(t)}$
where $\tau_1 = ( (1),{\bf 0}_{t-1} )$, $\tau_2 = ( (2,0),{\bf 0}_{t-2} ) $, 
$\tau_3 = ( (1,1),   {\bf 0}_{t-2} )$,
$\tau_4 = ( (3,0,0), {\bf 0}_{t-3} )$,
$\tau_5 = ( (2,1,0), {\bf 0}_{t-3} )$,
$\tau_6 = ( (1,1,1), {\bf 0}_{t-3} )$,
 and so on. 
We will consider quantum codes supported on Dicke states represented by the partitions of a suitably chosen integer $w$.
We then construct a matrix $A$ with rows labeled by the AD errors and columns labeled by Dicke states. 
In the paragraphs that follow, we will describe the structure of this matrix.

We now define the matrix elements of $A$.
They are
\begin{align}
a_{i,j} = \<\widetilde{{\bf q}_j} | A_{\tau_i, n}  ^\dagger A_{\tau_i, n}  | \widetilde{{\bf q}_j} \> 
\frac{1}{\gamma^{\wt(\tau_i)} (1-\gamma)^{n-\wt(\tau_i )}  },
\label{eq:aij}
\end{align}
where ${\bf q}_j$ are vectors with weight equal to $n$.
We can arrange these matrix elements into a matrix $A$, with the rows labeled by the AD errors, and the columns labeled by the quantum code's basis elements. 
The indices $i=1,\dots,\overline p(t)$ label the AD errors,
 and the indices $j=1,\dots, c$ label the Dicke states that the quantum code to be designed will be supported on. 
Writing down this matrix explicitly, we have
\begin{align}
A = 
\begin{pmatrix}
a_{1,1} & \dots  & a_{1,c} \\ 
\vdots & \dots & \vdots \\
a_{\overline p(t), 1} & \dots  & a_{\overline p(t) , c} \\ 
\end{pmatrix} 
=
\begin{pmatrix}
- {\bf a}_{1}^T - \\ 
\vdots \\
- {\bf a}_{\overline p(t)}^T -\\ 
\end{pmatrix} 
\label{eq:A}.
\end{align}
What is important about the matrix $A$ is that its properties will be used to 
design a PI constant-excitation quantum code that corrects AD errors.
For this to be possible, it is important that $A$ is independent of $\gamma$,
which indeed is the case because of the normalization condition in (\ref{eq:aij}).
Properties of the code will then be inferred from the nullity of $A$.

Independently from the matrix $A$, we can define basis states for a PI constant-excitation quantum code.
Our PI constant-excitation quantum code is thus defined only by the partitions labeling the Dicke states on which it is supported, and a real vector ${\bf x}$.
We represent the basis states of this quantum code in terms of linear combinations of Dicke states 
labeled by the partitions ${\bf q}_1,\dots, {\bf q}_{c}$ that all have the same weight equal to $n$,
and a non-zero real column vector ${\bf x} = (x_1 , \dots, x_c)^T$ 
such that $x_1+\dots + x_c =0$.
The basis states of our quantum code are
\begin{align}
|0_L\> &=  
\frac{1}{\sqrt x}\left( \sqrt{x_1^+}|\widetilde{{\bf q}_1}\> + \dots + \sqrt{x_c^+}|\widetilde{{\bf q}_c}\>  \right) 
\notag\\
|1_L\> &=  
\frac{1}{\sqrt x}\left( \sqrt{x_1^-}|\widetilde{{\bf q}_1}\> + \dots + \sqrt{x_c^-}|\widetilde{{\bf q}_c}\>  \right) 
\label{eq:code-defi}
\end{align}
where $x_i^+ = \max\{x_i,0\}$,
$x_i^- = \max\{-x_i,0\}$ and
$x = x_1^+ + \dots + x_c^+$. 

Roughly speaking, the matrix $A$ can be made to encapsulate the KL quantum error correction criterion with respect to the quantum code that we have defined in (\ref{eq:code-defi}).
More precisely, when a certain distance criterion holds and when the nullity of $A$ is at least one,
there are non-trivial solutions of the linear system of equations $A{\bf x} = 0$ for which 
$x_1 + \dots +x_c = 0$. 
This allows the derivation of a PI constant-excitation quantum code that corrects $t$ AD errors. 
This is our main result, and we state it in the following theorem.
\begin{theorem}
\label{thm:nullity-2-code}
Let $w, u$ and $t$ be positive integers and let $A$ be a matrix with matrix elements given by (\ref{eq:aij}). Let $Q=P(w)$ with $Q_u$ be given by (\ref{eq:Qu-defi}) and $\widetilde Q_u$ given by (\ref{def:widetilde Qu}).
If $d(\widetilde {Q_u}) \ge 2t+1$ and if the nullity of $A$ is at least one, 
then there exists a permutation-invariant constant-excitation quantum code that corrects $t$ AD errors using $uw$ total excitations.
Moreover, such a quantum code can be derived from (\ref{eq:code-defi});
the logical codewords (\ref{eq:code-defi}) derived from any non-zero vector ${\bf x}$ in the nullspace of $A$ will span such a quantum code.
\end{theorem}

Because bounding the nullity of $A$ is crucial in demonstrating that the quantum code as defined by (\ref{eq:code-defi}) corrects $t$ AD errors,
we will proceed to count the number of sets of linearly dependent rows in $A$ to obtain such a bound.
We use the fact that the rows in $A$ that correspond to $\kappa$ AD errors are linearly dependent. 
This arises because of the combinatorial identity in Proposition \ref{prop:const}.
This idea extends to certain submatrices of $A$ to demonstrate more linearly dependent rows.
The following lemma tells us how some rows of $A$ are linearly dependent,
where ${\bf a}_i^T$ denotes the $i$-th row of $A$.
\begin{lemma}
\label{lem:linear-dependence}
Let $c_i$ denote the number of ways to permute $(\tau_{i}| {\bf 0}_{n-n_i})$, where $n_i$ is the number of components in $\tau_i$.
For all $k=1, \dots, t$, the sum of the rows of $A$ 
 corresponding to the errors that induce $k$ photon losses sum to 
$c_{\overline  p (k-1) + 1}  {\bf a}_{\overline  p (k-1) + 1}^T
+ \dots +
c_{\overline  p (k)} {\bf a}_{\overline p(k)}^T =  \binom{n}{k}  {\bf 1}_c^T$.
\end{lemma}
\begin{proof}
Now note that $A_{\tau_i, n}  ^\dagger A_{\tau_i, n}$ is a diagonal matrix.
The number of ways to permute ${\bf q}_j$ is $|\widetilde{{\bf q}_j}|$. Hence the number of elements of the symmetric group that leave ${\bf q}_j$ invariant is $n! / |\widetilde{{\bf q}_j}|$.
Hence 
\begin{align}
\sum_{{\bf x} \in \widetilde{{\bf q}_j}} 
\< {\bf x} | A_{\tau_i, n}  ^\dagger A_{\tau_i, n} | {\bf x} \> 
=
\frac{1}{n! /| \widetilde{{\bf q}_j}| }
\sum_{\pi \in S_n} 
\<{\bf q}_j | \pi ^\dagger	 A_{\tau_i, n}  ^\dagger A_{\tau_i, n}  \pi | {\bf q}_j \> .
\end{align}
From this, 
\begin{align}
a_{i,j} \gamma^{\wt(\tau_i)} (1-\gamma)^{n-\wt(\tau_i)} = 
\<\widetilde{{\bf q}_j} | A_{\tau_i, n}  ^\dagger A_{\tau_i, n}  | \widetilde{{\bf q}_j} \>
&=
\frac{1}{|\widetilde{{\bf  q}_j}| }
\sum_{{\bf x} \in \widetilde{{\bf q}_j}} 
\< {\bf x} | A_{\tau_i, n}  ^\dagger A_{\tau_i, n} | {\bf x} \> \notag\\
&= \frac{1}{n!}
\sum_{\pi \in S_n} 
\< {\bf q}_j | \pi ^\dagger A_{\tau_i, n}  ^\dagger A_{\tau_i, n}  \pi | {\bf q}_j \>,
\end{align}
where $S_n$ denotes the matrix representation of the symmetric group that permutes the $n$ modes.
Using the definition of $c_i$, it follows that 
\begin{align}
\sum_{i = \overline p(k-1) + 1}^{\overline p(k)}
\frac{1}{n!/c_i}
\sum_{\pi \in S_n} 
\< {\bf q}_j | \pi ^\dagger A_{\tau_i, n}  ^\dagger A_{\tau_i, n}  \pi | {\bf q}_j \>
&=
\sum_{
\substack{
{\bf y} = (y_1, \dots, y_n) \in \mathbb N^n \\
y_1+ \dots + y_n = k \\
}}
\< {\bf q}_j | A_{\bf y}  ^\dagger A_{\bf y}  | {\bf q}_j \>
= \gamma^k(1-\gamma)^{n-k} \binom{n}{k} ,
\end{align}
where the last equality follows from Proposition \ref{prop:const}.
From this, it follows that 
\begin{align}
\sum_{i = \overline p(k-1) + 1}^{\overline p(k)}
c_i a_{i,j} =  \binom{n}{k},
\end{align}
and hence
\begin{align}
\sum_{i = \overline p(k-1) + 1}^{\overline p(k)}
c_i {\bf a}_{i}
= \binom{n}{k}  {\bf  1}_c  .\label{eq:rowsum}
\end{align}
\end{proof}
To better understand the ramification of Lemma \ref{lem:linear-dependence}, we explain the structure of the rows of $A$ in greater detail. The rows in $A$ are labeled by integer partitions corresponding to the AD errors. 
The first row of $A$ correspond to 1 AD errors. The second row and third corresponds to 2 AD errors with corresponding partitions given by (2,0) and (1,1) respectively.
The fourth, fifth and sixth rows corresponds to 3 AD errors with corresponding partitions given by (3,0,0) and (2,1,0) and (1,1,1) respectively.
Then Lemma \ref{lem:linear-dependence} implies the following.
\begin{enumerate}
\item One photon loss: The first row of $A$ is proportional to a vector of ones.
\item Two photon losses: A linear combination of the second and third rows of $A$ with positive integer coefficients is proportional to a vector of ones.  
\item Three photon losses: A linear combination of the fourth, fifth and sixth rows of $A$ with positive integer coefficients  is proportional to a vector of ones.
\end{enumerate}
Certain subsets of rows in $A$ are hence linearly dependent according to Lemma \ref{lem:linear-dependence},
namely the rows labeled by elements from  $\{\overline p(k-1) +1, \dots , \overline p(k)\}$ for every positive integer $k$.

By employing a different type of counting argument, one can note that different subsets of rows within $A$ are linearly dependent. This is given by the following Proposition.
\begin{proposition}\label{prop}
For any $h=1,2\dots, t-1$, the row vector ${\bf a}_{\bar p(h)}^T$ in $A$ is linearly dependent on the rows in $A$ where $w$ AD errors occur on at least $h$ modes, for every $w = h+1,\dots ,t$. 
\end{proposition}
\begin{proof}
The crux of the proof arises from the fact that $c_i a_{i,j}$ is a non-negative integer with a combinatoric interpretation.
Namely, we can interpret ${\bf q}_j$ as a column label in $A$ that specifies a list of distinguishable bins that all together contain $n$ indistinguishable balls.
Let $\wt(\tau_i)$ denote the sum of the components in the $\tau_i$.
Then we can interpret $c_i a_{i,j}$ as the number of ways of picking $\wt(\tau_i)$ balls such that the number of balls contributed by individual bins that conform to $\tau_i$.
\newline
For simplicity, denote $b_{i,j} = c_i a_{i,j}$ as entries of a matrix $B$ with row vectors ${\bf b}_i^T$. Also denote set of rows indices of $A$ for which $w$ AD errors occur on exactly $m$ modes to be 
\[I_{w,m} := \{i|wt(\tau_i) = w, \tau_i \mbox{ afflicts $m$ modes}\}.\]
To establish the proposition, we need to show that for any $h = 1, 2, . . . , t-1$ and 
$w > h$, the row vector ${\bf b}_{\bar p(h)}^T$  is a linear combination of ${\bf b}_i^T$, where the row indices belong to the set
\[I_{w,h} \cup \dots \cup I_{w,w}.\]
A crucial observation is now the following: given any partition ${\bf q}_j$ of $n$, ${\bf b}_{\bar p(h),j}$ is the number of ways of picking $h$ balls from $h$ different bins.
There is another way to calculate ${\bf b}_{\bar p(h),j}$ with an overcounting argument, by first considering too many AD errors, and then counting the number of ways to remove AD errors to get just the right number and configuration. 
To be precise, we first pick $w > h$ balls from at least $h$ bins, and then pick $h$ balls from the $w$ selected. Mathematically this reads
\[
\binom {n-h}{w-h} b_{\bar p(h),j} = \sum_{i \in \cup_{m=h}^w I_{w,h}} b_{i,j} d_{i,h},
\]
where $d_{i,h}$ is the number of ways of picking $h$ balls out of exactly $h$ bins from those given by $\tau_i$.
Note that $d_{i,h}$ and $\binom {n-h}{w-h}$ do not depend on $j$. This establishes the lemma.
\end{proof}

By identifying sets of linearly dependent rows of $A$ and counting the number of non-intersecting linearly dependent sets, one can obtain a lower bound on the nullity of $A$, 
from which a lower bound on the number of basis states can be obtained. This in turns implies that whenever the inequality (\ref{eq:inequality}) is satisfied, then we have a PI constant-excitation quantum code that corrects $t$ AD errors.
\begin{corollary}
\label{coro}
Let $w$ and $t$ be positive integers such that $p(w) + \binom t 2  \ge   p(1) + \dots + p(t)$ and $w \ge 2$. 
Then there is a permutation-invariant constant-excitation quantum code with $w(t+1)$ total excitations and which corrects $t$ AD errors.
\end{corollary}
\begin{proof}
Let us construct the matrix $A$ with columns labeled by the Dicke states labeled by $Q_u$ and rows labeled by AD errors of weight from 1 to $t$,
 where $Q = P(w)$ and $u =  t+1$. 
 
 We first show that $d(\widetilde Q_u) \ge 2(t+1)$, so that the distance criterion in Theorem \ref{thm:nullity-2-code} holds.
One can see this for the following reason. The minimum distance of any set of non-negative vectors of fixed length is trivially at least 2. 
Hence the minimum distance of $\widetilde Q$ is at least $2u$.
The minimum distance between vectors from $\widetilde Q$ and the ones vector is obviously 
$w(u-1) + (uw-w) = 2w(u-1)$ which is at least $2u$ whenever $w \ge 2$. 
 
 Let us denote the rank of $A$ and the nullity of $A$ by rank($A$) and nullity($A$) respectively.
Now the rank of $A$ is equal to its row rank, which is the number of its linearly independent rows. 
We will see that the matrix $A$ in fact has many linearly dependent rows, and hence its row rank is strictly less than the number of its rows. 
 More precisely, the sets of its rows which correct $\kappa$ AD errors for $\kappa =1,\dots, t$ are linearly dependent according to Lemma \ref{lem:linear-dependence}. 
 The case for $\kappa = 1$ is trivial, because there is only one row of $A$ that corresponds to $\kappa= 1$.
 Now define $L_\kappa = \{\overline p(\kappa-1) + 1, \dots , \overline p(\kappa) \}$.
 When $\kappa \ge 2$, the sets $L_\kappa$ of labels for the dependent row vectors of $A$ have cardinality at least two, and we can eliminate one dependent row from each $L_\kappa$, which leads to an elimination of $t-1$ rows.
According to Proposition \ref{prop}, the row of $A$ where $h$ AD errors afflict exactly $h$ modes is linearly dependent with the rows in $A$ that afflict $w$ AD errors in at least $h$ modes, for every $w > h$.
By setting $h=2,3,\dots$, we can eliminate another $t-2, t-3, \dots$ rows. The total number of linearly dependent rows is thus at least $1+\dots + (t-1) = \binom t 2$.
 
Now the dimension of the domain of $A$ is the number of its columns, which is $p(w)+1$. 
The rank-nullity theorem states that the nullity of $A$ is precisely $p(w)+1-{\rm rank}(A)$.
From an upper bound of the rank of $A$, we can obtain a lower bound on the nullity of $A$. 
We have seen from the previous paragraph that the rank of $A$ is at most the number of its rows minus $\binom t 2$.
Hence the rank-nullity theorem implies that nullity($A) \ge p(w) +1 - (\overline p(t) -\binom t 2 )$.
It follows that for the nullity of $A$ to be at least 1, it suffices to require
$p(w) +1 - (\overline p(t) - \binom t 2  ) \ge 1$. 
Theorem \ref{thm:nullity-2-code} then implies that we can use $A$ to construct a PI constant-excitation quantum code that corrects $t$ AD errors.
\end{proof}
\subsection{Proof of Theorem \ref{thm:nullity-2-code} }
First we show that the non-deformation condition with respect to the AD error of weight zero holds.
Notice that for any quantum state $|\psi\>$ that is a superposition of computational basis states each of weight $k$, 
Proposition \ref{prop:A-dag-A} implies that $\<\psi| (A_0^{\otimes n}) ^\dagger  A_0^ {\otimes n}|\psi\>  
= (1-\gamma)^k .$
Thus the non-deformation condition for the AD error of weight zero trivially holds.

Now we will demonstrate that the orthogonality conditions of the KL quantum error correction criterion are satisfied because of the distance criterion imposed.
While this has been proved in \cite{CLY97}[Theorem 2], 
we briefly state the underlying reason for this.
An AD error of weight $\kappa$ changes the weight of computation basis states by $\kappa$. The KL quantum error correction criterion involves taking the inner product of states both afflicted by AD errors of weight at most $t$. Thus, if the vectors underlying the computational basis states form a set of distance at least $2t+1$, 
all the orthogonal quantum error correction criterions will hold. 
Moreover, using the simple fact that $x_i = x_i^+ - x_i^-$, the non-deformation quantum error correction criterion for AD errors of weight from 1 to $t$
with respect to the code (\ref{eq:code-defi}) will be equivalent to the constraints
\begin{align}
\sum_{j=1}^c	x_j \< \widetilde{{\bf q}_j } |  A_{\tau_i, n}  ^\dagger A_{\tau_i ,n} | \widetilde{{\bf q}_j } \> 
&=
\<0_L|  A_{\tau_i, n}  ^\dagger A_{\tau_i, n} |0_L\> 
-
\<1_L|  A_{\tau_i, n}  ^\dagger A_{\tau_i, n} |1_L\>  = 0
\label{eq:proof-no-deform}.
\end{align}
Clearly these constraints are equivalent to the system of linear equations $A {\bf x} = {\bf 0}$.
But we still have to show is that there is a non-zero {\bf x} such that $x_1 + \dots +x_c = 0$ and  $A {\bf x} = {\bf 0}$.

We first show that if $A{\bf x} = 0$ has a non-trivial solution for ${\bf x}$, then $x_1+ \dots + x_c = 0$.
To see this, note that Lemma \ref{lem:linear-dependence} implies that the first row of $A$ is a vector of ones.
Hence $A{\bf x} = 0$ implies that 
${\bf 1}_{c}^T {\bf x} = 0$ which implies that $x_1 + \dots + x_c = 0$.
Therefore if the nullity of $A$ is at least one, the code as defined by (\ref{eq:code-defi}) exists and 
the non-deformation quantum error correction criterions for AD errors of weights from 1 to $t$ hold.
Since we have argued in the previous paragraph how all the orthogonality quantum error correction criterions hold and the non-deformation quantum error correction criterion for the AD of weight zero holds, 
all the KL quantum error correction criterions hold,
and the code as defined by (\ref{eq:code-defi}) corrects $t$ AD errors. \qed

 \section{Explicit code constructions}
\setcounter{theorem}{0}
In this section, we demonstrate how one can make use of the results in the previous section to construct quantum codes.
We illustrate briefly a recipe in which quantum codes may be found.
Suppose first that we wish to construct a quantum code that corrects $t$ AD errors.
Then we will pick some integer $w$, and set $Q = \overline P(w)$,
so that $Q$ is the set of all integer partitions of $1,\dots,w$.
We will next construct the set $Q_u$ for a suitable choice of an integer $u$.
The basis states of our quantum codes are then labeled by the elements of $Q_u$. 
With $Q_u$ and integer partitions labeling the different types of AD errors, we can construct a matrix $A$ as given in (\ref{eq:aij}).
Then we define our quantum codes based on the vectors that we find in the nullspace of $A$.

\begin{example}[Constant energy code correcting 1 AD error {\cite{WaB07}}] \label{example-t=1}
Consider $t=1, w=1, Q = P(w)$ with $u = 3$, so that the number of modes is $n=uw=3$. 
Then $Q_u = \{(3,0,0),(1,1,1)\}$. Obviously $d(\widetilde Q_u) = 4 \ge 2t+1$.
\begin{align}
A &= \begin{pmatrix}
a_{1,1} & a_{1,2} \\
\end{pmatrix}
\end{align}
where
\begin{align}
a_{1,1} 
&= \<\widetilde{(3,0,0)}|A_{(1,0,0)} ^\dagger A_{(1,0,0)} | \widetilde{(3,0,0)} \>
 / (\gamma (1-\gamma)^2)
\notag\\ 
&=
\frac{1}{3} \<(3,0,0)|A_{(1,0,0)} ^\dagger A_{(1,0,0)} | (3,0,0) \>  / (\gamma (1-\gamma)^2).
\end{align}
Note that 
\begin{align}
a_{1,2} &= \< (1,1,1) |A_{(1,0,0)} ^\dagger A_{(1,0,0)} | (1,1,1) \>   / (\gamma (1-\gamma)^2)  =1 .
\end{align} 
Hence $
A =  
\begin{pmatrix}
1 & 1 \\
\end{pmatrix}$.
Note that we can obtain the same result for $A$ from (\ref{eq:rowsum}) because since the number of ways to permute $(1,0,0)$ is 3, $3A = 3 {\bf 1}_2 ^T.$
The vector $x= (1,-1)$ clearly lies within the nullspace of $A$, and hence we can derive from 
(\ref{eq:code-defi}) the quantum code spanned by
\begin{align}
|0_L\> &= \frac{1}{ \sqrt  3  } \left(
|(3,0,0)\> + |(0,3,0)\>  + |(0,0,3)\>
\right)  \label{eq:t=1-L0} \\
|1_L\> &=  |(1,1,1)\>. \label{eq:t=1-L1}
\end{align}
Since all of the requirements of Theorem \ref{thm:nullity-2-code} are satisfied for $t=1$, the code spanned by (\ref{eq:t=1-L0}) and (\ref{eq:t=1-L1}) is a constant energy code which also corrects 1 AD error and which is permutation-invariant. 
This is also precisely Wasilewski and Banaszek's 3 mode code \cite{WaB07}.
\end{example}

\begin{example}[Constant energy code correcting 2 AD errors]\label{example-t=2}
Consider $t=2, w=2, Q = P(w)$ with $u = t+1 = 3$, so that the number of modes is $n=uw=6$. 
Then 
\begin{align}
Q_u = \{(6,0,0,0,0,0),(3,3,0,0,0,0),(1,1,1,1,1,1)\}.
\end{align}
Obviously $d(\widetilde Q_u) = 6 \ge 2t+1$. Also,
\begin{align}
A &= \begin{pmatrix}
1 & 1 & 1 \\
\frac 5 2 & 1 & 0 \\
0 & \frac 3 5 & 1 \\
\end{pmatrix}.
\end{align}
Now note that $A (\frac 2 5 , -1 , \frac 3 5) = 0,$ and hence we can derive from 
(\ref{eq:code-defi}) the quantum code spanned by
\begin{align}
|0_L\> &= \sqrt{\frac{2}{5}} 
|\widetilde{(6,0,0,0,0,0)}\>
+ \sqrt{\frac{3}{5}} |1\> ^{\otimes 6} 
,
\label{eq:t=2-L0} \\
|1_L\> &=  |\widetilde{(3,3,0,0,0,0)}\>. \label{eq:t=2-L1}
\end{align}
\end{example}

\begin{example}[Constant energy code correcting 3 AD errors with 12 excitations]\label{example-t=3,N=12}
Consider $t=3, w=3, Q = P(w)$ with $u = t+1 = 4$, so that the number of modes is $n=uw=12$. 
Then 
\begin{align}
Q_u = \{ 
( (12) | {\bf 0}_{11} ),
( (8,4)| {\bf 0}_{10} ),
( (4,4,4)| {\bf 0}_{9} ),
{\bf 1}_{12} 
\}.
\end{align}
Obviously $d(\widetilde Q_u) = 8 \ge 2t+1$. 
We now proceed to evaluate the matrix elements of $A$.
The first row of $A$ is a vectors of ones.
By considering only the matrix elements of $A_2^\dag A_2$,
the second row of $A$ is equal to 
\[
r_2 =\left( 
\frac{\binom {12} 2}{12} ,
 \frac{ 11 \binom{8}{2}+11\binom{4}{2}  }{2 \binom {12}{2} },
 \frac{  \binom{11}{2} \binom{4}{2}  }{ \binom {12}{3} },
0
  \right)^T=
  \left( \frac{11}{2}, \frac{17}{6}, \frac{3}{2}, 0    \right)^T
  .\]
By considering only the matrix elements of $ A_1^\dag A_1 \otimes A_1^\dag A_1$,
the third row of $A$ is equal to 
\[r_3 = \left( 
0 ,
 \frac{  2 \binom{8}{1} \binom{4}{1}  }{2 \binom {12}{2} },
 \frac{  \binom{10}{1} 4^2  }{ \binom {16}{3} },
 1
  \right)^T \notag\\
=
\left(0,\frac{16}{33},\frac{8}{11},1\right)^T .
  \]
As one can see, the first, second row and the third row are linearly dependent because $ \binom{12}{1} r_2 + \binom{12}{2} r_3 = \binom{12}{2} {\bf 1}_{12}$ as implied by Lemma \ref{lem:linear-dependence}.

Next we proceed to evaluate the fourth, fifth and sixth rows of $A$.
By considering only the matrix elements of $ A_3^\dag A_3$,
the fourth row of $A$ is equal to 
\[r_4= \left( 
\frac{\binom {12} 3}{12} ,
 \frac{ 11 \binom{8}{3}+11\binom{4}{3}  }{2 \binom {12}{2} },
 \frac{  \binom{11}{2} \binom{4}{3}  }{ \binom {12}{3} },
0
  \right)^T=
  \left(55/3,5,1,0\right)^T.
  \]
By considering only the matrix elements of $ A_2^\dag A_2 \otimes A_1^\dag A_1$,
the fifth row of $A$ is equal to 
\[r_5 = \left( 
0 ,
 \frac{  \binom{8}{2} \binom{4}{1}  + \binom{8}{1} \binom{4}{2} }{2 \binom {12}{2} },
 \frac{  10 \binom{ 4}{ 2} \binom 4 1  }{ \binom {12}{3} },
 0
  \right)^T \notag\\
=
\left(0,\frac{40}{33},\frac{12}{11},0\right)^T .
  \]
By considering only the matrix elements of $ \left(A_1^\dag A_1 \right) ^{\otimes  3}$,
the sixth row of $A$ is equal to 
\[r_6 = \left( 
0 ,
0,
 \frac{  \binom 4 1 ^3 }{ 3\binom {16}{3} },
 1
  \right)^T \notag\\
=
\left(0,0,\frac{16}{55},1\right)^T .
  \]
Clearly, we have $\binom {12}{1} r_4 + 2\binom{12}{2} r_5+ \binom{12}{3} r_6  = \binom{12}{3} {\bf 1}_{12}^T$, and hence the first, fourth, fifth and sixth rows are linearly dependent. We now get
\begin{align}
A &= \begin{pmatrix}
1 & 1 & 1 & 1 \\
\frac{11}{2} &\frac{17}{6}&  \frac{3}{2} &  0 \\
0&  \frac {16}{33} & \frac 8 {11} &   1 \\
 55/3 & 5 &   1 &0 \\
0 &\frac{40}{33} & \frac{12}{11} & 0 \\
 0 & 0  &\frac{16}{55} & 1 \\ 
\end{pmatrix}.
\end{align}
The matrix rank of $A$ is 3,
and the null space of $A$ is spanned by
$(-21/32, 99/32, -55/16, 1)$.
Thus we have $A (-21, 99, -110 , 32) = 0$. 
From this we can derive from 
(\ref{eq:code-defi}) the quantum code spanned by
\begin{align}
|0_L\> &= 
\frac{1}{\sqrt{131}}
\left(
\sqrt{99} |\widetilde{((8,4)| {\bf 0}_{10} )}\>
+ 
\sqrt{32} |1\> ^{\otimes 12}
\right),
\label{eq:t=3-L0-N12} \\
|1_L\> &=  
\frac{1}{\sqrt{131}}
\left(
\sqrt{21} |\widetilde{( (12) | {\bf 0}_{11} )}\>
+ 
\sqrt{110} |\widetilde{( (4,4,4)| {\bf 0}_{9} )}\>
\right)
. \label{eq:t=3-L1-N12}
\end{align}
\end{example}

\begin{example}[Constant energy code correcting 4 AD errors] \label{example-t=4,N=20}
Consider $t=4, w=4, Q = P(w)$ with $u = t+1 = 5$, so that the number of modes is $n=uw=20$. 
The Dicke states are specified by
\begin{align}
Q_u = \{ 
( (20) | {\bf 0}_{19} ),
( (15,5)| {\bf 0}_{18} ),
( (10,10)| {\bf 0}_{18} ),
( (10,5,5)| {\bf 0}_{17} ),
( (5,5,5,5)| {\bf 0}_{16} ),
{\bf 1}_{20}  \}.
\end{align}
Then the matrix $A$ is given by 
\begin{align}
A = \left(\begin{array}{cccccc} 1 & 1 & 1 & 1 & 1 & 1\\ \frac{19}{2} & \frac{23}{4} & \frac{9}{2} & \frac{13}{4} & 2 & 0\\ 0 & \frac{15}{38} & \frac{10}{19} & \frac{25}{38} & \frac{15}{19} & 1\\ 57 & \frac{93}{4} & 12 & 7 & 2 & 0\\ 0 & \frac{135}{76} & \frac{45}{19} & \frac{75}{38} & \frac{30}{19} & 0\\ 0 & 0 & 0 & \frac{25}{114} & \frac{25}{57} & 1\\ \frac{969}{4} & \frac{137}{2} & 21 & 11 & 1 & 0\\ 0 & \frac{485}{76} & \frac{120}{19} & \frac{75}{19} & \frac{30}{19} & 0\\ 0 & \frac{105}{19} & \frac{405}{38} & \frac{100}{19} & \frac{60}{19} & 0\\ 0 & 0 & 0 & \frac{425}{684} & \frac{50}{57} & 0\\ 0 & 0 & 0 & 0 & \frac{125}{969} & 1 \end{array}\right),
\end{align}
with rank 5 and nullity 1.
The nullspace is spanned by $\left(\begin{array}{cccccc} \frac{84}{125} & - \frac{456}{125} & - \frac{152}{125} & \frac{1368}{125} & - \frac{969}{125} & 1 \end{array}\right)$.
From this we can derive from 
(\ref{eq:code-defi}) the quantum code spanned by
\begin{align}
|0_L\> &= 
\frac{1}{\sqrt{1577}}
\left(
\sqrt{84} |\widetilde{((20)| {\bf 0}_{19} )}\>
+ 
\sqrt{1368} |\widetilde{((10,5,5)| {\bf 0}_{17} )}\>
+
\sqrt{125} |1\> ^{\otimes 20}
\right),
\label{eq:t=4-L0-N20} \\
|1_L\> &=  
\frac{1}{\sqrt{1577}}
\left(
\sqrt{456} |\widetilde{( (15,5) | {\bf 0}_{18} )}\>
+ 
\sqrt{152} |\widetilde{( (10,10) | {\bf 0}_{18} )}\>
+ 
\sqrt{969} |\widetilde{( (5,5,5,5) | {\bf 0}_{16} )}\>
\right)
. \label{eq:t=4-L1-N20}
\end{align}
\end{example}

\begin{example}[Constant energy code correcting 5 AD errors] \label{example-t=5,N=30}
Consider $t=5, w=5, Q = P(w)$ with $u = t+1 = 6$, so that the number of modes is $n=uw=30$. 
The Dicke states are specified by
\begin{align}
Q_u = \{ &
( (30) | {\bf 0}_{30} ),
( (24,6)| {\bf 0}_{28} ),
( (18,12)| {\bf 0}_{28} ),
( (18,6,6)| {\bf 0}_{27} ),\notag\\
&\quad
( (12,12,6)| {\bf 0}_{27} ),
( (12,6,6,6)| {\bf 0}_{26} ),
( (6,6,6,6,6)| {\bf 0}_{26} ),
{\bf 1}_{30}  \}.
\end{align}
Then the matrix $A$ is given by 
\begin{align}
A = \left(\begin{array}{cccccccc} 1 & 1 & 1 & 1 & 1 & 1 & 1 & 1\\ \frac{29}{2} & \frac{97}{10} & \frac{73}{10} & \frac{61}{10} & \frac{49}{10} & \frac{37}{10} & \frac{5}{2} & 0\\ 0 & \frac{48}{145} & \frac{72}{145} & \frac{84}{145} & \frac{96}{145} & \frac{108}{145} & \frac{24}{29} & 1\\ \frac{406}{3} & \frac{1022}{15} & \frac{518}{15} & \frac{428}{15} & \frac{46}{3} & \frac{28}{3} & \frac{10}{3} & 0\\ 0 & \frac{336}{145} & \frac{504}{145} & \frac{426}{145} & \frac{456}{145} & \frac{378}{145} & \frac{60}{29} & 0\\ 0 & 0 & 0 & \frac{162}{1015} & \frac{216}{1015} & \frac{54}{145} & \frac{108}{203} & 1\\ \frac{1827}{2} & \frac{3547}{10} & \frac{237}{2} & 103 & \frac{67}{2} & 18 & \frac{5}{2} & 0\\ 0 & \frac{2104}{145} & \frac{2292}{145} & \frac{1792}{145} & \frac{280}{29} & \frac{180}{29} & \frac{80}{29} & 0\\ 0 & \frac{276}{29} & \frac{3366}{145} & \frac{321}{29} & \frac{2112}{145} & \frac{243}{29} & \frac{150}{29} & 0\\ 0 & 0 & 0 & \frac{729}{1015} & \frac{972}{1015} & \frac{1269}{1015} & \frac{270}{203} & 0\\ 0 & 0 & 0 & 0 & 0 & \frac{96}{1015} & \frac{48}{203} & 1\\ \frac{23751}{5} & 1417 & 312 & 286 & 53 & 27 & 1 & 0\\ 0 & \frac{10686}{145} & \frac{1521}{29} & \frac{1248}{29} & \frac{606}{29} & \frac{333}{29} & \frac{60}{29} & 0\\ 0 & \frac{276}{29} & \frac{3366}{145} & \frac{321}{29} & \frac{2112}{145} & \frac{243}{29} & \frac{150}{29} & 0\\ 0 & \frac{1196}{29} & \frac{14586}{145} & \frac{1040}{29} & 44 & 18 & \frac{200}{29} & 0\\ 0 & 0 & 0 & \frac{1053}{406} & \frac{594}{145} & \frac{1593}{406} & \frac{675}{203} & 0\\ 0 & 0 & 0 & 0 & 0 & \frac{312}{1015} & \frac{120}{203} & 0\\ 0 & 0 & 0 & 0 & 0 & 0 & \frac{144}{2639} & 1 \end{array}\right)
,
\end{align}
with rank 7 and nullity 1.
The nullspace is spanned by 
\begin{align}
\left(\begin{array}{cccccccc} - \frac{21505}{31104} & \frac{135575}{31104} & \frac{39875}{15552} & - \frac{55825}{3888} & - \frac{25375}{2592} & \frac{5075}{144} & - \frac{2639}{144} & 1 \end{array}\right).\notag
\end{align}
From this we can derive from (\ref{eq:code-defi}) the quantum code spanned by
\begin{align}
|0_L\> = 
\frac{1}{\sqrt{1342629}}&
\left(
\sqrt{135575} |\widetilde{((24,6)| {\bf 0}_{28} )}\>
+
\sqrt{79750} |\widetilde{((18,12)| {\bf 0}_{28} )}\>
\right.
\notag	\\
&\left.\quad+
\sqrt{1096200} |\widetilde{((12,12,6)| {\bf 0}_{27} )}\>
+
\sqrt{31104} |1\> ^{\otimes 30}
\right),
\label{eq:t=5-L0-N30} \\
|1_L\> =  
\frac{1}{\sqrt{1342629}}&
\left(
\sqrt{21505} |\widetilde{((30)| {\bf 0}_{29} )}\>
+ 
\sqrt{446600} |\widetilde{( (18,6,6) | {\bf 0}_{27} )}\>\right.
\notag\\
&\left.\quad
+ 
\sqrt{304500} |\widetilde{( (12,12,6) | {\bf 0}_{27} )}\>
+ 
\sqrt{570024} |\widetilde{( (6,6,6,6,6) | {\bf 0}_{25} )}\>
\right)
. \label{eq:t=5-L1-N30}
\end{align}
\end{example}

To illustrate the fact that $A$ can potentially have a nullspace larger than 1, we consider in the following a 3 AD quantum code with 16 excitations.

\begin{example}[Constant energy code correcting 3 AD errors] \label{example-t=3,N=16}
Consider $t=3, w=4, Q = P(w)$ with $u = t+1 = 4$, so that the number of modes is $n=uw=16$. 
Then 
\begin{align}
Q_u = \{ 
( (16) | {\bf 0}_{15} ),
( (12,4)| {\bf 0}_{14} ),
( (8,8)| {\bf 0}_{14} ),
( (8,4,4)| {\bf 0}_{13} ),
( (4,4,4,4)| {\bf 0}_{12} ),
{\bf 1}_{16} 
\}.
\end{align} 
Obviously $d(\widetilde Q_u) = 8 \ge 2t+1$. 
We now proceed to evaluate the matrix elements of $A$.
The first row of $A$ is a vectors of ones.
By considering only the matrix elements of $A_2^\dag A_2$,
the second row of $A$ is equal to 
\[
r_2 =\left( 
\frac{\binom {16} 2}{16} ,
 \frac{ 15 \binom{12}{2}+15\binom{4}{2}  }{2 \binom {16}{2} },
 \frac{  15 \binom{8}{2}  }{ \binom {16}{2} },
 \frac{  \binom{15}{2} \binom{8}{2} + 2 \binom{15}{2} \binom{4}{2}  }{ 3\binom {16}{3} },
 \frac{  \binom{15}{3} \binom{4}{2}  }{ \binom {16}{4} },
0
  \right)^T=
  \left(\frac{15}{2},\frac{9}{2},\frac{7}{2},\frac{5}{2},\
  \frac{3}{2},0\right)^T
  .\]
By considering only the matrix elements of $ A_1^\dag A_1 \otimes A_1^\dag A_1$,
the third row of $A$ is equal to 
\[r_3 = \left( 
0 ,
 \frac{  2 \binom{12}{1} \binom{4}{1}  }{2 \binom {16}{2} },
 \frac{  \binom{8}{1}\binom 8 1  }{ \binom {16}{2} },
 \frac{ 2 \binom{8}{1}\binom 4 1 (14)  + \binom 4 1 \binom 4 1 (14) }{ 3\binom {16}{3} },
 \frac{  \binom{4}{1}\binom 4 1 \binom{14}{2}  }{ \binom {16}{4} },
 1
  \right)^T \notag\\
=
\left(0,\frac{2}{5},\frac{8}{15},\frac{2}{3},\frac{4}{5},1\right)^T .
  \]
As one can see, the first, second row and the third row are linearly dependent because $ \binom{16}{1} r_2 + \binom{16}{2} r_3 = \binom{16}{2} {\bf 1}_{16}$ as implied by Lemma \ref{lem:linear-dependence}.

Next we proceed to evaluate the fourth, fifth and sixth rows of $A$.
By considering only the matrix elements of $ A_3^\dag A_3$,
the fourth row of $A$ is equal to 
\[r_4= \left( 
\frac{\binom {16} 3}{16} ,
 \frac{ 15 \binom{12}{3}+15\binom{4}{3}  }{2 \binom {16}{2} },
 \frac{  15 \binom{8}{3}  }{ \binom {16}{3} },
 \frac{  \binom{15}{2} \binom{8}{3} + 2 \binom{15}{2} \binom{4}{3}  }{ 3\binom {16}{3} },
 \frac{  \binom{15}{3} \binom{4}{3}  }{ \binom {16}{4} },
0
  \right)^T=
  \left(35,14,7,4,1,0\right)^T
  \]
By considering only the matrix elements of $ A_2^\dag A_2 \otimes A_1^\dag A_1$,
the fifth row of $A$ is equal to 
\[r_5 = \left( 
0 ,
 \frac{  \binom{12}{2} \binom{4}{1}  + \binom{12}{1} \binom{4}{2} }{2 \binom {16}{2} },
 \frac{  \binom{8}{2} \binom 8 1  }{ \binom {16}{2} },
 \frac{  \left(\binom{8}{2}\binom 4 1 + \binom 8 1 \binom  4 2 \right) (14)  
 + \binom 4 2 \binom 4 1 (14) }{ 3\binom {16}{3} },
 \frac{  \binom{4}{2}\binom 4 1 \binom{14}{2}  }{ \binom {16}{4} },
 0
  \right)^T \notag\\
=
\left(0,\frac{7}{5},\frac{28}{15},\frac{23}{15},\frac{6}
   {5},0\right)^T .
  \]
By considering only the matrix elements of $ \left(A_1^\dag A_1 \right) ^{\otimes  3}$,
the sixth row of $A$ is equal to 
\[r_6 = \left( 
0 ,
0,
0,
 \frac{ 3 \binom{8}{1}\binom 4 1 \binom 4 1}{ 3\binom {16}{3} },
 \frac{  \binom{4}{1}^3  \binom{13}{1}  }{ \binom {16}{4} },
 1
  \right)^T \notag\\
=
\left(0,0,0,\frac{8}{35},\frac{16}{35},1\right)^T .
  \]
Clearly, we have $\binom {16}{1} r_4 + 2\binom{16}{2} r_5+ \binom{16}{3} r_6  = \binom{16}{3} {\bf 1}_{16}^T$, and hence the first, fourth, fifth and sixth rows are linearly dependent. We now get
\begin{align}
A &= \begin{pmatrix}
1 & 1 & 1 & 1 & 1 & 1\\
\frac{15}{2} &\frac{9}{2}&  \frac{7}{2}&  \frac 5 2& \frac 3 2 &  0 \\
0&  \frac 2 5& \frac 8 {15}& \frac 2 3 & \frac 4 5 & 1 \\
 35 & 14 & 7 & 4 & 1 &0 \\
0 & \frac{7}{5} & \frac{28}{15} & \frac{23}{15} & \frac{6}{5} &0  \\
 0 & 0 & 0 &\frac{8}{35} & \frac{16}{35} & 1 \\ 
\end{pmatrix}.
\end{align}
The matrix rank of $A$ is 3,
and the null space of $A$ is spanned by
\begin{align}
\begin{pmatrix}
-17 / 12 \\ 115 / 24 \\ 0 \\ -35 / 8 \\ 0 \\ 1 \\
\end{pmatrix}
,
\begin{pmatrix}
-1 / 3\\ 4 / 3\\ 0\\ -2\\ 1 \\ 0 \\ 
\end{pmatrix}
,
\begin{pmatrix}
1 / 3 \\ -4 / 3 \\ 1 \\ 0\\ 0\\ 0\\
\end{pmatrix}.
\end{align}
Thus we trivially have for instance 
$A (1/3, -4/3, 1, 0, 0, 0) = 0$. 
From this we can derive from 
(\ref{eq:code-defi}) the quantum code spanned by
\begin{align}
|0_L\> &= 
\sqrt{\frac{1}{4}} 
|\widetilde{((16)| {\bf 0}_{15} )}\>
+ \sqrt{\frac{3}{4}} |\widetilde{((8,4,4)| {\bf 0}_{13} )}\>
,
\label{eq:t=3-L0} \\
|1_L\> &=  |\widetilde{((12,4)| {\bf 0}_{14} )}\>
. \label{eq:t=3-L1}
\end{align}
\end{example}

In Table \ref{mytable}, we present parameters constant-excitation permutation-invariant quantum codes that can be constructed using our methodology for $t=1, \dots, 10$.
 \begin{table}[]
\centering
\label{mytable}
\begin{tabular}{|l|l|l|}
\hline
\multicolumn{1}{|l|}{$t$} & \multicolumn{1}{l|}{$N$}& \multicolumn{1}{l|}{$(t+1)^2$} \\
 \hline
1                       &    3      & 4  \\
2                       &    6      & 9  \\
3                       &   12     & 16 \\
4                       &   20     & 25   \\ 
5                       &   30     & 36   \\ 
6                       &   $49^{*}$   & 49   \\ 
7                       &   $72^{*}$    & 64   \\ 
8                       &   $90^{*}$   & 81   \\ 
9                       &   $120^{*}$  & 100  \\ 
10                     &   $143^{*}$   & 121 \\ 
\hline
\end{tabular}
\caption{Table of code parameters. The first column are values for $t$, the number of AD errors the quantum code can correct. The second column, $N$ is the total excitation number for our constant-excitation quantum code that corrects $t$ AD errors and is permutation-invariant. 
The third column is $(t+1)^2$, which is the total excitation number of Bergmann and van Loock's codes \cite{BvL16}, which are not permutation-invariant.
The numbers for $N$ marked with an asterisk are obtained from \eqref{eq:inequality} and are likely not to be smallest possible.
On the other hand, the numbers for $N$ marked without an asterisk have their codes given explicitly in the examples we provided.}
\end{table}
 
 \begin{figure}
  \includegraphics[width=\textwidth]{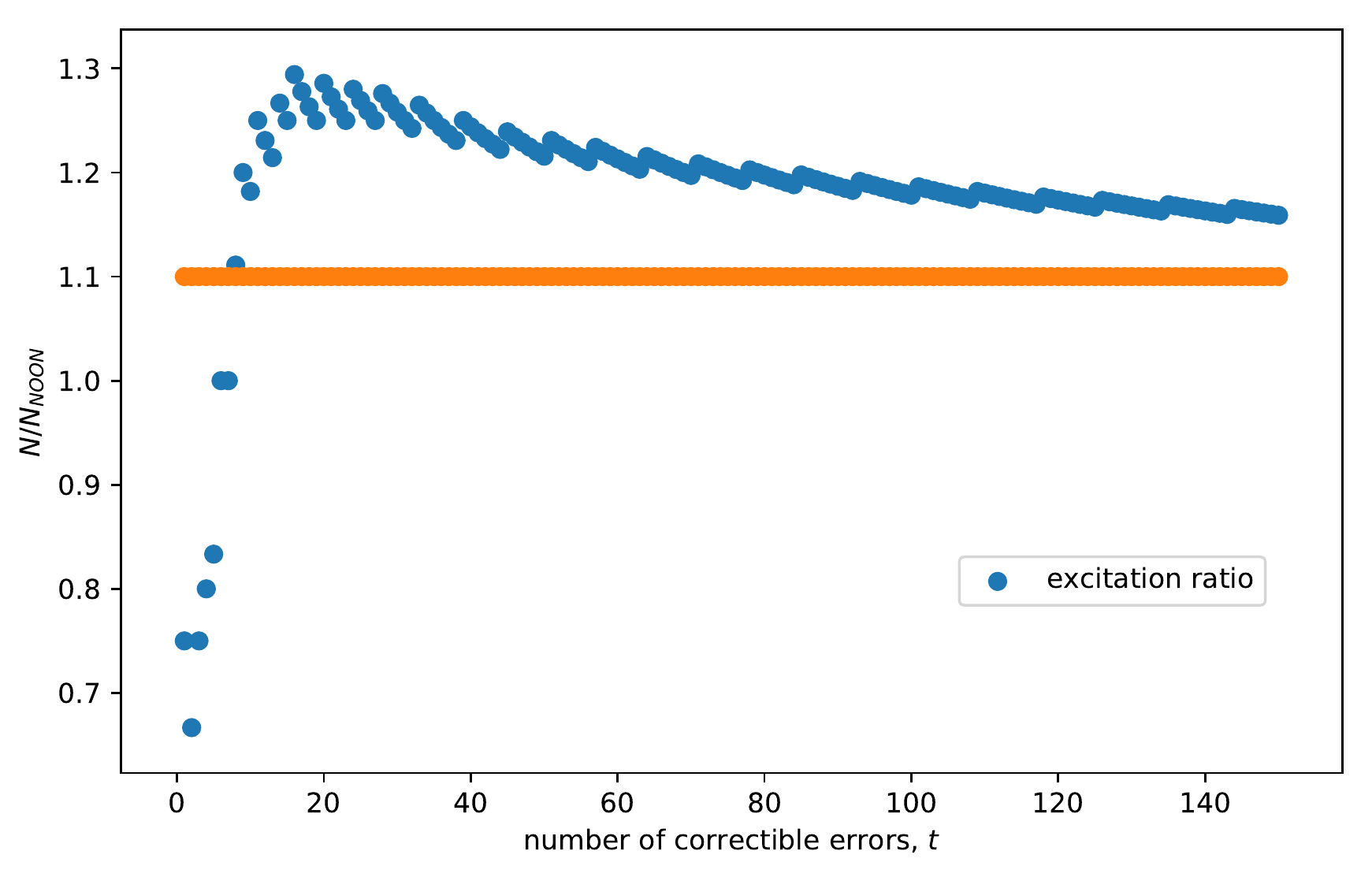}
  \caption{
We denote as $N_{NOON} = (t+1)^2$ and $N$ as the number of excitations needed to correct $t$ AD errors using 
the NOON codes of  Bergmann and van Loock codes \cite{BvL16} and our permutation-invariant constant-excitation codes respectively. The first 5 data points corresponding to our explicit code constructions outperform the NOON codes, and the remaining data points suggest that $N \le 1.3(t+1)^2$. \label{fig:bounds} }
\end{figure}

\section{Discussions} \label{sec:discussions}
In this paper, we study codes that lie within the decoherence-free subspace of certain Hamiltonians, while also exhibiting the ability to reverse the effects of some amplitude damping errors. 
We focus on the Hamiltonian that is a sum of quantum harmonic oscillators of identical frequencies, because it can describe the quantum bus used by superconducting qubits. 
Since permutations may unexpectedly occur, 
it is also advantageous for quantum codes to also exhibit permutation-invariance.
Here in this paper, we present a method where vectors from the nullspace of a matrix $A$ can give rise to constant-excitation PI codes that correct AD errors, which are naturally immune to the natural dynamics of both the quantum harmonic oscillators and also arbitrary permutations.

To label the bases on which our codes are constructed on, we have used the all ones vector along with the set $Q=P(w)$. 
One might wonder if one could construct quantum codes using our method with $Q$ a strict subset of $P(w)$. 
This is useful if certain Dicke states are unphysical to implement for example. One can see from 
Example \ref{example-t=3,N=16} that if the initial $Q$ is chosen as a suitable strict subset of $P(w)$, the derived quantum code would be exactly the same. In particular, we supply Example \ref{example-t=3,N=16} to illustrate the fact that by choosing $N$ to be larger than the minimum required to correct $t$ errors, we can have some flexibility in choosing which Dicke states our code is to be supported on.

The codes considered here are non-degenerate. The effect of different AD errors on the permutation-invariant code are distinct. To see how this happens explicitly, let us consider the example of Wasilewski and Banaszek's excitation 3 code.
Consider the AD errors $A_{(1,0,0)}$ and $A_{(0,0,1)}$.
Notice that
\begin{align}
A_{(1,0,0)}|0_L\> = \frac{1}{\sqrt 3} \sqrt 3|(2,0,0)\> \sqrt{\gamma (1-\gamma)^2}  \label{q1}\\
A_{(0,0,1)}|0_L\> = \frac{1}{\sqrt 3} \sqrt 3|(0,0,2)\> \sqrt{\gamma (1-\gamma)^2} \label{q2}.
\end{align} 
One can see that the effect of $A_{(1,0,0)}$ and $A_{(0,0,1)}$ on the logical zero codeword is not the same because 
\eqref{q1} is not equal to \eqref{q2}.

Regarding the literature on constant-excitation quantum codes that are not necessarily PI,
we have also improved on their construction when 2,3,4 and 5 AD errors are to be corrected. 
Namely, our codes require only 6,12,20 and 30 total excitations to correct 2,3,4 and 5 AD errors respectively.
In contrast, Chuang, Leung and Yamamoto constructed codes correcting 2 and 3 AD errors using 9 and 16 total excitations \cite{CLY97},
while Bergmann and van Loock constructed codes correcting $t$ AD errors using $(t+1)^2$ total excitations \cite{BvL16}.
We construct the best constant-excitation quantum codes that can correct between 2 to 5 AD errors, as the number of excitations needed to correct $t$ errors for $t=2,3,4,5$ is $t(t+1)$, 
which is less than the $(t+1)^2$ total excitations previously needed.
Explicit code constructions using our method are supplied in Example \ref{example-t=2} and Example \ref{example-t=3,N=12} to correct 2 and 3 AD errors using 6 and 12 total excitations respectively.
Our construction can also be seen as a generalization of Wasilewski and Banaszek's constant-excitation PI code correcting 1 AD error \cite{WaB07} to PI constant-excitation quantum codes that can correct an arbitrary number of AD errors.
Asymptotically, it also appears that the number of excitations needed to correct $t$ AD errors for our codes exhibits the same behavior as Bergmann and van Loock's codes, because both require $O(t^2)$ excitations (see Figure \ref{fig:bounds}).

Bosonic codes with constant excitation number were previously c+onsidered by the authors in \cite{CLY97}, where they established some fundamental properties of such codes. 
While the quantum codes considered here are indeed a subfamily of the quantum codes considered in \cite{CLY97}, they do not follow trivially from \cite{CLY97} for the following reasons.
First, in \cite{CLY97}, quantum codes satisfying the non-deformation Knill-Laflamme quantum error correction criterion and certain orthogonality conditions are constructed, with explicit construction algorithms for 1 and 2 AD errors.
However for at least 3 AD errors, one could only rely on brute force search as there was no systematic way to generate such constant-excitation quantum codes.
In contrast, we see in this paper how constant excitation quantum codes for any number of AD errors can be generated.
Second, to numerically find codes in \cite{CLY97}, the authors check if a system of linear equations corresponding to the non-deformation conditions is satisfied. In the language used in this paper, they essentially determine the nullity of the matrix $A$, but it was unclear how one could construct the quantum code using this information. 
We fill this research gap in this paper by demonstrating how the code can be explicitly constructed from the nullspace of the matrix $A$, therefore enriching the theory of this nascent field.
Third, in \cite[Eq. (5.3)]{CLY97}, the authors give an inequality that is asymptotically equivalent to ours in (1.2), which tells us when constant excitation codes exist. However, since arbitrary constant excitation codes were considered, it is {\em a priori} unclear if the inequality \cite[Eq. (5.3)]{CLY97} holds when the additional constraint of permutation-invariance on the codes is imposed. Here, we show that this is in fact possible, and slightly improve their inequality. While their inequality states that constant excitation codes exist when
\begin{align}
p\left(\frac{N}{t+1}\right) \ge ( p(0) + p(1) + \dots p(t) ) + 2, 
\end{align}
we say that permutation-invariant constant excitation codes exist when
\begin{align}
p\left(\frac{N}{t+1}\right) + \binom t 2 \ge (  p(1) + \dots p(t) ).
\end{align}
Just like in the case of \cite{CLY97}, this analytical bound is not tight, as we demonstrate for small values of $t$.

One limitation of this paper is that our constructed PI constant-excitation quantum codes only correct against AD errors. It would be advantageous to study when our codes can also correct against a fixed number of arbitrary errors. 
Another limitation is that we have only provided a theoretical structure of our PI constant-excitation quantum codes; the practicality of implementing our codes has yet to be fully addressed. 
We expect that techniques in preparation of Dicke states and quantum cellular automata to be useful with regards to this issue. However, this lies beyond the scope of the current paper, where our focus lies primarily only on the mathematical structure of PI constant-excitation quantum codes that correct AD errors.

In summary, we prove the existence of constant-excitation quantum codes that not only correct any number of AD errors, but are invariant under any permutations. 
When certain distance criterion are satisfied, we also provide a new method for obtaining quantum codes  by finding vectors that lie within the nullspace of a matrix.

\section{Acknowledgments}
The author thanks the referees for valuable comments.
The author acknowledges support from the National Research Foundation and Ministry of Education, Singapore. This material is based on research funded in part by the Singapore National Research Foundation under NRF Award NRF-NRFF2013-01 and the U.S. Air Force Office of Scientific Research under AOARD grant FA2386-18-1-4003.
This work was supported by the EPSRC (Grant No. EP/M024261/1).
This work was also supported by the QCDA project (Grant No. EP/R043825/1)) which has received funding from the QuantERA ERANET Cofund in Quantum Technologies implemented within the European Union’s Horizon 2020 Programme.

\bibliography{mybib}{}

\begin{thebibliography}{10}

\bibitem{shor1994algorithms}
P.~W. Shor, ``Algorithms for quantum computation: Discrete logarithms and
  factoring,'' in {\em Foundations of Computer Science, 1994 Proceedings., 35th
  Annual Symposium on}, pp.~124--134, Ieee, 1994.

\bibitem{shor1999polynomial}
P.~W. Shor, ``Polynomial-time algorithms for prime factorization and discrete
  logarithms on a quantum computer,'' {\em SIAM review}, vol.~41, no.~2,
  pp.~303--332, 1999.

\bibitem{lloyd1996universal}
S.~Lloyd, ``{Universal quantum simulators},'' {\em Science}, vol.~273,
  pp.~1073--1078, 1996.

\bibitem{BB84}
C.~H. Bennett and G.~Brassard, ``{Quantum cryptography: Public key distribution
  and coin tossing},'' in {\em Proceedings of IEEE International Conference on
  Computers, Systems and Signal Processing}, vol.~175, New York, 1984.

\bibitem{Eke91}
A.~K. Ekert, ``Quantum cryptography based on {B}ell's theorem,'' {\em Phys.
  Rev. Lett.}, vol.~67, pp.~661--663, Aug 1991.

\bibitem{KnL97}
E.~Knill and R.~Laflamme, ``{Theory of quantum error-correcting codes},'' {\em
  Phys. Rev. A}, vol.~55, pp.~900--911, Feb. 1997.

\bibitem{blais2004cavity}
A.~Blais, R.-S. Huang, A.~Wallraff, S.~M. Girvin, and R.~J. Schoelkopf,
  ``Cavity quantum electrodynamics for superconducting electrical circuits: An
  architecture for quantum computation,'' {\em Physical Review A}, vol.~69,
  no.~6, p.~062320, 2004.

\bibitem{majer2007coupling}
J.~Majer, J.~Chow, J.~Gambetta, J.~Koch, B.~Johnson, J.~Schreier, L.~Frunzio,
  D.~Schuster, A.~Houck, A.~Wallraff, {\em et~al.}, ``Coupling superconducting
  qubits via a cavity bus,'' {\em Nature}, vol.~449, no.~7161, p.~443, 2007.

\bibitem{gu2017microwave}
X.~Gu, A.~F. Kockum, A.~Miranowicz, Y.-x. Liu, and F.~Nori, ``Microwave
  photonics with superconducting quantum circuits,'' {\em Physics Reports},
  2017.

\bibitem{heinze2018universal}
M.~Heinze and R.~Koenig, ``Universal uhrig dynamical decoupling for bosonic
  systems,'' {\em arXiv preprint arXiv:1810.07117}, 2018.

\bibitem{leung-thesis}
D.~W. Leung, {\em {Towards Robust Quantum Computation}}.
\newblock PhD thesis, Stanford University, 2000.

\bibitem{CLY97}
I.~L. Chuang, D.~W. Leung, and Y.~Yamamoto, ``{Bosonic quantum codes for
  amplitude damping},'' {\em Phys. Rev. A}, vol.~56, p.~1114, 1997.

\bibitem{ouyang2015permutation}
Y.~Ouyang and J.~Fitzsimons, ``Permutation-invariant codes encoding more than
  one qubit,'' {\em Phys. Rev. A}, vol.~93, p.~042340, Apr 2016.

\bibitem{pazy2012semigroups}
A.~Pazy, {\em Semigroups of linear operators and applications to partial
  differential equations}, vol.~44.
\newblock Springer Science \& Business Media, 2012.

\bibitem{LNCY97}
D.~W. Leung, M.~A. Nielsen, I.~L. Chuang, and Y.~Yamamoto, ``{Approximate
  quantum error correction can lead to better codes},'' {\em Phys. Rev. A},
  vol.~56, p.~2567, 1997.

\bibitem{cochrane1999macroscopically}
P.~T. Cochrane, G.~J. Milburn, and W.~J. Munro, ``Macroscopically distinct
  quantum-superposition states as a bosonic code for amplitude damping,'' {\em
  Physical Review A}, vol.~59, no.~4, p.~2631, 1999.

\bibitem{SSSZ09}
P.~W. Shor, G.~Smith, J.~A. Smolin, and B.~Zeng, ``{High Performance
  Single-Error-Correcting Quantum Codes for Amplitude Damping},'' {\em IEEE
  Transactions on Information Theory}, vol.~57, pp.~7180--7188, Oct. 2011.

\bibitem{DGJZ10}
R.~Duan, M.~G. Ji, Zhengfeng, and B.~Zeng, ``{Multi-Error Correcting Amplitude
  Damping Codes},'' {\em ISIT}, 2010.

\bibitem{BinomialCodes2016}
M.~H. Michael, M.~Silveri, R.~T. Brierley, V.~V. Albert, J.~Salmilehto,
  L.~Jiang, and S.~M. Girvin, ``New class of quantum error-correcting codes for
  a bosonic mode,'' {\em Phys. Rev. X}, vol.~6, p.~031006, Jul 2016.

\bibitem{jackson2016concatenated}
T.~Jackson, M.~Grassl, and B.~Zeng, ``Concatenated codes for amplitude
  damping,'' in {\em Information Theory (ISIT), 2016 IEEE International
  Symposium on}, pp.~2269--2273, IEEE, 2016.

\bibitem{li2017cat}
L.~Li, C.-L. Zou, V.~V. Albert, S.~Muralidharan, S.~Girvin, and L.~Jiang, ``Cat
  codes with optimal decoherence suppression for a lossy bosonic channel,''
  {\em Physical review letters}, vol.~119, no.~3, p.~030502, 2017.

\bibitem{grassl2018quantum}
M.~Grassl, L.~Kong, Z.~Wei, Z.-Q. Yin, and B.~Zeng, ``Quantum error-correcting
  codes for qudit amplitude damping,'' {\em IEEE Transactions on Information
  Theory}, vol.~64, no.~6, pp.~4674--4685, 2018.

\bibitem{WaB07}
W.~Wasilewski and K.~Banaszek, ``Protecting an optical qubit against photon
  loss,'' {\em Phys. Rev. A}, vol.~75, p.~042316, Apr 2007.

\bibitem{BvL16}
M.~Bergmann and P.~van Loock, ``Quantum error correction against photon loss
  using {NOON} states,'' {\em Phys. Rev. A}, vol.~94, p.~012311, Jul 2016.

\bibitem{Rus00}
M.~B. Ruskai, ``{Pauli Exchange Errors in Quantum Computation},'' {\em Phys.
  Rev. Lett.}, vol.~85, pp.~194--197, July 2000.

\bibitem{PoR04}
H.~Pollatsek and M.~B. Ruskai, ``{Permutationally invariant codes for quantum
  error correction},'' {\em Linear Algebra and its Applications}, vol.~392,
  no.~0, pp.~255--288, 2004.

\bibitem{ouyang2014permutation}
Y.~Ouyang, ``{P}ermutation-invariant quantum codes,'' {\em Phys. Rev. A},
  vol.~90, no.~6, p.~062317, 2014.

\bibitem{OUYANG201743}
Y.~Ouyang, ``Permutation-invariant qudit codes from polynomials,'' {\em Linear
  Algebra and its Applications}, vol.~532, pp.~43 -- 59, 2017.

\bibitem{gottesman2002introduction}
D.~Gottesman, ``An introduction to quantum error correction,'' in {\em
  Proceedings of Symposia in Applied Mathematics}, vol.~58, pp.~221--236, 2002.

\end{thebibliography}
\bibliographystyle{ieeetr}

\end{document}